\documentclass[12pt]{article} %***
\usepackage[sectionbib]{natbib}
\usepackage{array,epsfig,fancyheadings,rotating}
\usepackage[]{hyperref}  %<----modified by Ivan
\usepackage{sectsty, secdot}
\sectionfont{\fontsize{12}{14pt plus.8pt minus .6pt}\selectfont}
\renewcommand{\theequation}{\thesection\arabic{equation}}
\subsectionfont{\fontsize{12}{14pt plus.8pt minus .6pt}\selectfont}
\usepackage{amsthm}
\usepackage{pdflscape}

\usepackage{subfigure}
\usepackage{graphicx}
\usepackage{amssymb,amstext,amsmath}
\usepackage{float}
\usepackage{bbm}
\usepackage{bm}
\usepackage{multirow}
\usepackage{makecell}
\usepackage{color}
\usepackage{verbatim}

\newcommand{\alphab}{\mathbf{\alpha}}

\newcommand{\x}{{\bf x}}
\newcommand{\X}{{\bf X}}

\newtheorem{theorem}{Theorem}

\newtheorem{corolario}{\em Corollary}

\newtheorem{nota}{\em Remark}
\def\Item$#1${\item $\displaystyle#1$
\hfill\refstepcounter{equation}(\theequation)}

\begin{document}

%%%%%%%%%%%%%%%%%%%%%%%%%%%%%%%%%%%%%%%%%%%%%%%%%%%%%%%%%%%%%%%%%%%%%%%%%%%%%%%%%%%%%%%%%%%%%%%%%%%%%%%%%%%%%%%%%%%%%%%%%%%%
%%%%%%%%%%%%%%%%%%%%%%%%%%%%%%%%%%%%%%%%%%%%%%%%%%%%%%%%%%%%%%%%%%%%%%%%%%%%%%%%%%%%%%%%%%%%%%%%%%%%%%%%%%%%%%%%%%%%%%%%%%%%

\renewcommand{\baselinestretch}{2}

\markright{ \hbox{\footnotesize\rm 
%{\footnotesize\bf 24} (201?), 000-000
}\hfill\\[-13pt]
\hbox{\footnotesize\rm
%\href{http://dx.doi.org/10.5705/ss.20??.???}{doi:http://dx.doi.org/10.5705/ss.20??.???}
}\hfill }

\markboth{\hfill{\footnotesize\rm MATEU SBERT AND V\'ICTOR ELVIRA} \hfill}
{\hfill {\footnotesize\rm FILL IN A SHORT RUNNING TITLE} \hfill}

\renewcommand{\thefootnote}{}
$\ $\par

%%%%%%%%%%%%%%%%%%%%%%%%%%%%%%%%%%%%%%%%%%%%%%%%%%%%%%%%%%%%%%%%%%%%%%%%%%%%%%%%%%%%%%%%%%%%%%%%%%%%%%%%%%%%%%%%%%%%%%%%%%%%

\fontsize{12}{14pt plus.8pt minus .6pt}\selectfont \vspace{0.8pc}
\centerline{\large\bf  GENERALIZING THE BALANCE HEURISTIC ESTIMATOR }
\vspace{2pt} \centerline{\large\bf IN MULTIPLE IMPORTANCE SAMPLING}
%\centerline{\large\bf  Generalizing the Balance Heuristic Estimator }
%\vspace{2pt} \centerline{\large\bf in Multiple Importance Sampling}
\vspace{.3cm} \centerline{Mateu Sbert$^\dagger$ and V\'ictor Elvira$^*$} \vspace{.2cm} %\centerline{\it \small{$^\dagger$College of Intelligence and Computing, Tianjin University, Tianjin, China} } 
 \vspace{0cm} \centerline{\it
 \small{$^\dagger$Informatics and Applications Institute Girona University, Girona, 17003, Spain} }\vspace{0cm}  \centerline{\it
\small{$^*$IMT Lille Douai \& CRIStAL Laboratory (UMR 9189), Lille, 59653, France} }\vspace{.55cm} \fontsize{9}{11.5pt plus.8pt minus
.6pt}\selectfont

%%%%%%%%%%%%%%%%%%%%%%%%%%%%%%%%%%%%%%%%%%%%%%%%%%%%%%%%%%%%%%%%%%%%%%%%%%%%%%%%%%%%%%%%%%%%%%%%%%%%%%%%%%%%%%%%%%%%%%%%%%%%

\begin{quotation}
\noindent {\it Abstract:} %\\
In this paper, we propose a novel and generic family of multiple importance sampling estimators. We first revisit the celebrated balance heuristic estimator, a widely used Monte Carlo technique for the approximation of intractable integrals. Then, we establish a generalized framework for the combination of samples simulated from multiple proposals. We show that the novel framework contains the balance heuristic as a particular case. In addition, we study the optimal choice of the free parameters in such a way the variance of the resulting estimator is minimized. A theoretical variance study shows the optimal solution is always better than the balance heuristic estimator (except in degenerate cases where both are the same). As a side result of this analysis, we also provide new upper bounds for the balance heuristic estimator. Finally, we show the gap in the variance of both estimators by means of five numerical examples.

\vspace{9pt}
\noindent {\it Key words:}
 Monte Carlo, importance sampling, balance heuristic, variance reduction. 
\par
\end{quotation}\par

\def\thefigure{\arabic{figure}}
\def\thetable{\arabic{table}}

\renewcommand{\theequation}{\thesection.\arabic{equation}}

\fontsize{12}{14pt plus.8pt minus .6pt}\selectfont

\setcounter{section}{0} %***
\setcounter{equation}{0} %-1

\lhead[\footnotesize\thepage\fancyplain{}\leftmark]{}\rhead[]{\fancyplain{}\rightmark\footnotesize}
%\thepage}%Put this line in Page 2

%%%%%%%%%%%%%%%%%%%%%%%%%%%%%%%%%%%%%%%
\section{Introduction}
\label{sec:intro}
%%%%%%%%%%%%%%%%%%%%%%%%%%%%%%%%%%%%%%%

Multiple importance sampling (MIS) is a Monte Carlo technique widely used in the literature of signal processing, computational statistics, and computer graphics for approximating complicated integrals. In its basic configuration, it works by drawing random samples from several proposal distributions (also called techniques) and weighting them appropriately in such a way an estimator built with the pairs of weighted samples is consistent. Since, the publication of \cite{Veach:1995:MIS}, the celebrated \emph{balance heuristic} estimator has been extensively used in the Monte Carlo literature, with an unprecedented success in the computer graphics industry.\footnote{Eric Veach has been awarded with several prizes because of his contributions in the MIS literature, where the Balance Heuristic is arguably the most relevant one.} In the balance heuristic method, different samples are simulated from each proposal and the traditional IS weight is assigned to each of them. Unlike the standard IS estimator, all the weighted samples are combined with an extra weighting, in such a way the resulting estimator typically shows a reduced variance. Its superiority in terms of variance w.r.t. other traditional combination schemes has been  recently shown in \cite{elvira2019generalized}, where a framework is established for sampling and weighting in MIS under equal number of samples per technique.
The balance heuristic, also called \emph{deterministic mixture} (\cite{Owen98safeand}), has been widely used in the literature of MIS. 
%In \cite{2015arXiv151103095E}, the balance heuristic is proved to  
Further efficient variance reduction techniques are proposed in \cite{ElviraMLB15,elvira2016multiple,elvira2016heretical} also in the context of MIS, still with equal counts from each technique. Provably better estimators
(\cite{SbertHavranSzirmay2016}) and heuristically better ones
(\cite{vrcaiHavranSbert14,SbertHavran2017}) have been
presented that use a different count of samples than equal count for all techniques. In~\cite{SbertHavranSzirmayElvira2018} it has been shown the relationship of a better count of samples with generalized weighted means. The balance heuristic is also present in most of successful adaptive IS (AIS) methods, see \cite{Cappe04,CORNUET12,martino2017layered,elvira2017improving,bugallo2017adaptive}, in particular in the case where all techniques are used to simulate the same number of samples. 

Interestingly, the balance heuristic has two properties in the assigned weights. First, all techniques appear at the denominator of the weight of a specific technique. Second, they appear in a form of a mixture, with coefficients proportional to the number of samples simulated from each technique. In this paper, we relax this constraint providing a generalized weighting/combining family of estimators that has the balance heuristic as a particular case. First, we show that it is possible to use a specific set of coefficients to decide the amount of samples per technique, and a different set of  coefficients to be applied as the importance weight. Second, we study four different cases fixing some of these coefficients (sampling and/or weighting), and we give the optimal solution for the rest of coefficients in such a way the variance of the MIS estimator is minimized. Note that, the novel estimator always outperforms the balance heuristic under the optimal choice of those coefficients. Third, we complete the theoretical work with three new upper bounds for the variance of the balance heuristic estimator. In five numerical examples we show that, under an adequate choice of parameters, the novel estimator outperforms the celebrated balance heuristic. These examples are also an interesting testbed for deepening in the study of the aforementioned novel upper bounds on the variance of the balance heuristic estimator. 

The rest of the paper is structured as follows. Section \ref{balanceheuristic} revisits the balance heuristic estimator. In Section \ref{anewfamily}, we propose the new family of estimators that generalizes the balance heuristic. We address four cases of special interest, depending on the number of samples simulated from each technique. Finally, we conclude the paper with five numerical examples in Section \ref{1D examples} and some conclusions in Section \ref{sec:conclusion}.

\section{Balance heuristic estimator}
\label{balanceheuristic}
The goal in IS is usually the estimation of the
value of integral $\mu = \int f(\x) {\mathrm d}\x $. In MIS, $n_i$ samples, $\{ \X_{i,j} \}_{j=1}^{n_i}$, are simulated from a set of available probability density functions (pdfs), $\{ p_i \}_{i=1}^n$. The MIS estimator introduced by Veach and Guibas~\cite{Veach:1995:MIS} is given by
\begin{equation}
\label{MISestimator1}
Z = \sum_{i=1}^{n} \frac{1}{n_i} \sum_{j=1}^{n_i} w_i(\X_{i,j}) \frac{f(\X_{i,j})} {p_i(\X_{i,j})},
\end{equation}
where $w_i(\x)$ is a weight function associated to the $i$-th proposal that fulfills both following conditions. First, the weights must sum up to one in all points of the domain where the value of the function is different from zero, i.e., 
%\begin{equation}
%\label{MISestimator2}
$\sum_{i=1}^{n} w_i(\x) =1$, $\forall x$ where $f(\x)\neq 0$,
%\end{equation}
%\cred{
Second, for all $x$ where $p_i(\x) = 0$, then $w_i(\x) =0$.

The \emph{balance heuristic} estimator is a particular case of Eq. \eqref{MISestimator1} where the weight function is given by
\begin{equation}
\label{balanceheuristic1}
w_i(\x) = \frac{n_i p_i(\x)}{\sum_{k=1}^n n_k p_k(\x)},
\end{equation}
which can be written too as
\begin{equation}
\label{balanceheuristic1bis1}
w_i(\x) = \frac{\alpha_i p_i(\x)}{\sum_{k=1}^n \alpha_k p_k(\x)},
\end{equation}
where $n_i= \alpha_i N$. Then, the estimator in Eq. \eqref{MISestimator1} becomes the balance heuristic or deterministic mixture estimator given by
\begin{eqnarray}
\label{MISestimator1bis1}
F &=& \sum_{i=1}^{n} \frac{1}{n_i} \sum_{j=1}^{n_i}  \frac{\alpha_i f(\X_{i,j})}{\sum_{k=1}^n \alpha_k p_k(\X_{i,j})}\\
&=& \frac{1}{N} \sum_{i=1}^{n}  \sum_{j=1}^{n_i}  \frac{f(\X_{i,j})}{\sum_{k=1}^n \alpha_k p_k(\X_{i,j})}.
\end{eqnarray}
\subsection{Interpretation of $F$ and general notation of the paper}
\label{sec_notation}
Note that, some authors interpret $F$ in Eq.~\eqref{MISestimator1bis1} as an estimator where the $N$ samples are simulated from the mixture $\psi_{\alphab} = \sum_{k=1}^n \alpha_k p_k(\x)$, which appears in the denominator of all importance weight, even if the selection of the number of samples $n_i$ per technique is deterministic (see \cite[Appendix 1]{elvira2019generalized}). The $n_i$ are then random variables with expected values $\alpha_i N$. In our framework, we fix deterministically $n_i$, as in the balance heuristic estimator (\cite{Veach:1995:MIS}), which is also called deterministic mixture scheme (\cite{Owen98safeand}), as opposition to the case where all samples are simulated from the mixture $\psi_{\alphab}$. Deterministic mixture sampling can be seen as a Rao-blackwellization which reduces the variance in the sampling but also in the IS estimators~(\cite{Owen98safeand,elvira2019generalized}). The estimator with random number of samples is sometimes called one-sample MIS estimator, while the deterministic number of samples is sometimes denoted as  multi-sample MIS estimator (\cite{SbertHavranSzirmay2016,Veach:PhD}). In this paper, we refer the former as random mixture estimator and the later as deterministic mixture estimator. All estimators, unless the opposite is clearly stated, use a deterministic selection of the number of samples per technique. Moreover, those estimators with the superindex $^1$ are versions of a specific estimator where the number of samples is normalized to $1$, e.g. $F^{1}$. In otherwords, even if the estimators require that all the numbers of samples per technique are $n_i \in \mathbb{R}$, we use this normalized estimators to denote the variance normalized to $1$ sample, which simplifies the comparison across estimators (for $N$ total samples, the variance of the estimator would be just the variance of $F^{1}$ divided by $N$).
In Table~\ref{table:estimators} we show the naming convention used in this paper.
\begin{table*}
\def\arraystretch{1.2}
\caption{\label{table:estimators} Naming convention for the multiple importance sampling
  estimators in this paper. We will drop the superindex ${1}$ from
  primary estimators when not strictly necessary.}
\begin{center}
\scriptsize{
\begin{tabular}{|l|l|}
 \hline
$Z$ & Generic deterministic (multi-sample) MIS estimator \\
 $Z^1$ & Generic deterministic (multi-sample) MIS estimator normalized to one sample\\
  ${\cal{Z}}$ & Generic randomized (one-sample) MIS estimator \\
   ${\cal{Z}}^1$ & Generic randomized (one-sample) MIS estimator, for number of samples equal to 1 \\
 ${{Z}}_{l1},{\cal{Z}}_{l1}$ &Deterministic and randomized, respectively, optimal lineal combination when weights are fixed and constant for each technique\\
  ${{Z}}_{l2}, {\cal{Z}}_{l2}$ &Deterministic and randomized, respectively, optimal lineal combination when weights are equal to sampling proportions\\
 ${{Z}}_{l3}, {\cal{Z}}_{l3}$ &Deterministic and randomized, respectively, optimal lineal combination when sampling proportions are fixed\\
$F$ & Balance heuristic multi-sample MIS estimator \\
$F^1$ & Balance heuristic multi-sample MIS estimator normalized to one sample\\
${\cal{F}}$ & Generalized balance heuristic one-sample MIS estimator \\
${\cal{F}}^1$ & Generalized balance heuristic one-sample MIS estimator, for number of samples equal to 1  \\
%${\mathcal F}$ & Balance heuristic one-sample MIS estimator \\
$G$ & Generalized balance heuristic multi-sample MIS estimator \\
 \hline
\end{tabular}
}
\vspace*{2mm}
\end{center}
\end{table*}
\subsection{Rationale}
\label{preliminaries}
In \cite[Theorems 9.2 and 9.4]{Veach:PhD}, 
 the relationship between the variances of $Z^1$ and its randomized version is discussed (i.e., the version where instead of deterministically selecting $n_i$, all the samples are directly simulated from $\sum_{k=1}^n \alpha_k p_k(\x)$). The same result is obtained in \cite{elvira2019generalized} in a broader variance analysis of MIS estimators. In particular, it can be shown that 
\begin{eqnarray}
\label{differenceofZ1}
V[{\cal Z}^1]-V[Z^1]= \sum_i \alpha_i {\mu'}_i^2 - \mu^2, 
\end{eqnarray}
where
 \begin{eqnarray}
\label{differenceofZ2}
{\mu'}_i = \frac{1}{\alpha_i} \int w_i(\x) f(\x) d\x,
\end{eqnarray}
and thus
\begin{eqnarray}
\label{differenceofZ3}
\sum_i \alpha_i {\mu'}_i=\mu.
\end{eqnarray}
From Eq.~\eqref{differenceofZ1} (see Appendix A) we have that
\begin{eqnarray}
\label{differenceofZ5}
V[Z^1] \le V[{\cal Z}^1] ,
\end{eqnarray}
and equality only happens (apart from the case when both variances $V[Z^1],V[{\cal Z}^1]$ are zero) when for all $i$ all ${\mu'_i}$ are equal. One example is given by taking in Eq.~\eqref{differenceofZ2} for all $i$, $w_i(\x)=w_i$ constant and $\alpha_i= w_i$, see the estimators $Z_{l2}$ and ${\cal{Z}}_{l2}$ in Appendix C. For the particular case when $\alpha_i = 1/n$ (see Appendix A),
\begin{eqnarray}
\label{differenceofZ6}
V[Z^1]-V[{\cal Z}^1] \le (n-1) \mu^2. 
\end{eqnarray}
Veach also proved (\cite{Veach:PhD}), Theorem 9.4, that the optimal weights for ${\cal Z}$, this is, the weight functions $w_i(\x)$ that minimize $V[{\cal Z}]$, are the balance heuristic ones,  Eq.~\eqref{balanceheuristic1bis1}, and thus in the optimal case ${\cal Z} \equiv {\cal F}$, where ${\cal F}$ is the random mixture estimator. This is, for any estimator ${\cal Z}$, we have that using the same distribution of samples, taking into account also Eq. \eqref{differenceofZ5} (see also \cite{SbertHavranSzirmay2016}), it always holds that 
\begin{eqnarray}
\label{differenceofZ6}
V[{F}]\le V[{\cal F}] \le V[{\cal Z}].
\end{eqnarray}
Eq.~\ref{differenceofZ6} will be used in Section~\ref{case1} to find new upper bounds for $V[{F}]$. 
Also in \cite{Veach:PhD}, Theorem 9.2, it is proved that the estimator that optimizes the second moment of $Z^1$ estimator, this is, $V[Z^1]+ \sum_i \alpha_i {\mu'}_i^2$, is the balance heuristic estimator. Thus, it seems clear that for improvement we have to look for a deterministic estimator, that should be a generalization of balance heuristic mixture estimator $F$. This will be done in next section.

\section{Generalized Multiple Importance Sampling Balance Heuristic  estimator}
\label{anewfamily}
Let us consider the estimator of Eq. \eqref{MISestimator1bis1}, where we relax the dependence between the number of samples $n_i$, and the associated coefficient $\alpha_i$, i.e., now $n_i = \beta_i N$, $\beta_i > 0$, $\sum_{i=1}^n \beta_i =1$, where in general $\alpha_i \ne \beta_i$ (otherwise, we recover $F$). We now define the estimator
 \begin{eqnarray}
 G&=&  \sum_{i=1}^n  \frac{\alpha_i}{n_i} \sum_{j=1}^{n_i}   \frac{f(\X_{i,j})}{\sum_{k=1}^n \alpha_k p_k(\X_{i,j})}\\
  &=& \frac{1}{N} \sum_{i=1}^n  \frac{\alpha_i}{\beta_i} \sum_{j=1}^{n_i}   \frac{f(\X_{i,j})}{\sum_{k=1}^n \alpha_k p_k(\X_{i,j})}.
 \end{eqnarray}
 Note that $G$ is a particular case of $Z$, with weights $w_i = \frac{\alpha_i p_i(\x)}{\sum_{k=1}^n \alpha_k p_k(\X_{i,j})}$ in Eq. \eqref{MISestimator1}. {Note that the  balance heuristic $F$ is a particular case of $G$, i.e., in general we do not impose the restriction of $\alpha_i = \frac{n_i}{N}$.}
 
\begin{theorem} 
For any set of weights $\{\alpha_i \}_{i=1}^n$, such as $\sum_{i=1}^n \alpha_i = 1$ and any set of weights $\{\beta_i \}_{i=1}^n$, such as $\sum_{i=1}^n \beta_i = 1$, $G$ is an unbiased estimator of $\mu$.
\end{theorem} 
\begin{proof}
Note that the estimator can be rewritten as $G = \sum_{i=1}^n \alpha_i G_i$, where
 \begin{eqnarray}
\label{anewfamily01}
G_i &=&  \frac{1}{n_i} \sum_{j=1}^{n_i}  \frac{f(\X_{i,j}))}{\sum_{k=1}^n \alpha_k p_k(\X_{i,j}))}.
\end{eqnarray}
 Note also that $G$ depends of two sets of parameters, $\{\alpha_i  \}_{i=1}^n, \{\beta_i \}_{i=1}^n$. In the particular case where $\beta_i = \alpha_i, \forall i$, the estimator $G$ becomes $F$. Let us first consider the case with $n_i = 1$. Then,
 \begin{eqnarray}
\label{anewfamily3}
G_i' &=&  \frac{f(\x)}{\sum_{k=1}^n \alpha_k p_k(\x)},
\end{eqnarray}
with expectation
 \begin{eqnarray}
\label{anewfamily2}
E[G_i'] =\int \frac{f(\x) p_i(\x)}{\sum_{k=1}^n \alpha_k p_k(\x)}d\x \equiv  \mu'_i,
\end{eqnarray}
and variance
\begin{eqnarray}
\label{anewfamily5}
{\sigma'}_i^2 &=&   \int \frac{f^2(\x) p_i(\x)}{(\sum_{k=1}^n \alpha_k p_k(\x))^2} d\x -  ( {\mu'_i})^2.
\end{eqnarray}

 The estimator $G$ is unbiased, since
  \begin{eqnarray}
\label{anewfamily1}
E[G] &=& \sum_{i=1}^n \alpha_i \mu'_i  =  \sum_i \alpha_i \int \frac{f(\x) p_i(\x)}{\sum_{k=1}^n \alpha_k p_k(\x)} d\x  \\
&=& \int \frac{f(\x) \sum_{i=1}^n \alpha_i p_i(\x)}{\sum_{k=1}^n \alpha_k p_k(\x)} d\x  \\ \nonumber
&=&  \int {f(\x)} d\x \equiv \mu.
 \end{eqnarray}
 \end{proof}
 
The variance of $G$ is given by
 \begin{eqnarray}
\label{anewfamily6}
V [G] =V \left[ \sum_{i=1}^n \alpha_i G_i \right] =  \sum_{i=1}^n \alpha_i^2 V [G_i]= \sum_{i=1}^n \frac{ \alpha_i^2 {\sigma'}_i^2 }{n_i} .
 \end{eqnarray}
For the sake of the theoretical analysis, we define $G^{1}$, a normalized version of $G$ with $N=1$ (see Section \ref{sec_notation}), with variance
 \begin{eqnarray}
\label{anewfamily6bis}
V [G^1] = \sum_{i=1}^n \frac{ \alpha_i^2 {\sigma'}_i^2 }{\beta_i}.
 \end{eqnarray}
Next we study four special cases of estimator $G^1$. 
\begin{nota}
We could also consider the one-sample estimator $\cal{G}$, randomized version of $G$. However, $\cal{G}$ is a particular case of the general estimator $\cal{Z}$, and we have seen in Section~\ref{preliminaries} that the optimal case for $\cal{Z}$ is when $\cal{Z}\equiv\cal{F}$, thus it only makes sense to consider the extension $G$ of the multi-sample estimator $F$. 
\end{nota}
\subsection{Case 1: $\alpha_i = \beta_i$, $\forall i$ }
\label{case1}
In this particular case, the estimator $G$ reverts to $F$. The variance is 
    \begin{eqnarray}
\label{anewfamily9}
V[F^1] = \sum_{i=1}^n \alpha_i {\sigma'}_i^2,
  \end{eqnarray}
by simple substitution in Eq. \eqref{anewfamily6bis}. 
 We aim at finding the optimal $\{\alpha_i^*\}_{i=1}^n$ such the variance of Eq. \eqref{anewfamily9} is minimized.
\begin{theorem}
\label{theorem_variance_F}
The optimal estimator $F^*$ in terms of variance is achieved when the following expression is equal $\forall j\in\{1,..,n\}$, %Minimum variance condition is when for all $j$, the following values are equal
\begin{eqnarray}
\label{optimalFcondition}
  {\sigma'}_j^2 +2 {\mu'_j}^2
 - 2\sum_{i=1}^n \alpha_i^* {\mu'_i} \int \frac{f(\x) p_i(\x) p_j(\x)}{(\sum_{k=1}^n \alpha_k^* p_k(\x))^2}d\x.
\end{eqnarray}
\end{theorem}
\begin{proof}
See Appendix B for a proof.
\end{proof}

\begin{theorem}
\label{theorem_bound_1}
If all the $n$ sampling techniques are unbiased, $V[F^1]$ is upper bounded by the following bounds
\begin{enumerate}
\item{\begin{equation}
\label{bound11}
{\cal{A}} (v_i;\alpha_i),
\end{equation}}
\item{\begin{equation} 
\label{bound21}
  {\cal{H}} (v_i;\alpha_i)  + \mu^2 \left( \frac{({\cal{H}}(v_i;\alpha_i))^2}{{\cal{H}}(v_i^2;\alpha_i)} -1 \right),  \end{equation} }
\item{  \begin{equation} 
\label{bound31}
 \left({\cal{H}} (\sqrt{v_i};\alpha_i)\right)^2  + 
\mu^2 \left( \frac{\left({\cal{H}} (\sqrt{v_i};\alpha_i)\right)^2}{{\cal{H}}(v_i;\alpha_i)} -1 \right), 
\end{equation}  }
\end{enumerate}
where ${\cal{A}} (v_i;\alpha_i)$ and $ {\cal{H}} (v_i;\alpha_i) $ are the arithmetic and harmonic weighted averages, respectively, and {$v_i$ denotes the variance of a single-proposal IS estimator with one sample simulated from technique $p_i$.}
\end{theorem}
\begin{proof}
The proofs for all bounds can be found in Appendix C. Note that the first bound was already introduced in \cite{SbertHavran2017}.
\end{proof}
From Eq.~\eqref{bound11} follows immediately
\begin{corolario}
$V[F^1]$ is upper bounded by $max_i\{v_i\}$.
\end{corolario}
\begin{proof}
\begin{equation}
\label{lessthanmaximumvi}
V[F^1] \le \sum_i \alpha_i v_i\le max_i\{v_i\} \sum_i \alpha_i = max_i\{v_i\}.
\end{equation}
\end{proof}
Observe that $\left({\cal{H}} (\sqrt{v_i};\alpha_i)\right)^2$ is the weighted power mean with power=-1/2. We remind that arithmetic and harmonic means are power means with power 1 and -1, respectively. Thus we have the inequalities
$${\cal{H}} (v_i;\alpha_i) \le \left({\cal{H}} (\sqrt{v_i};\alpha_i)\right)^2 \le {\cal{A}} (v_i;\alpha_i),
$$
but these inequalities do not have to hold for the bounds, as they contain additional terms.
For the case of biased techniques, we have the {following bounds}.
\begin{theorem}
\label{theorem_bound_biased_1}
{The three upper bounds for $V[F^1]$ hold:}
\begin{enumerate}
\item{\begin{equation}
\label{bound12}
{V[F^1] \leq} {\cal{A}} (v_i;\alpha_i)+ {\cal{A}} (\mu_i^2;\alpha_i)-\mu^2,
\end{equation}}
\item{\begin{equation} 
\label{bound22}  
{V[F^1] \leq} {\cal{H}} (v_i;\alpha_i)  + \frac{({\cal{H}}(v_i;\alpha_i))^2}{{\cal{H}}(\frac{v_i^2}{\mu_i^2};\alpha_i)} -\mu^2,  
\end{equation} }
\item{  \begin{equation}  
\label{bound32}
{V[F^1] \leq} \left({\cal{H}} (\sqrt{v_i};\alpha_i)\right)^2  + 
\frac{\left({\cal{H}} (\sqrt{v_i};\alpha_i)\right)^2}{{\cal{H}}(\frac{v_i}{\mu_i^2};\alpha_i)} -\mu^2.  
\end{equation}  }
where $\mu_i$ is the expected value of integral $\mu$ when sampling with technique $i$, and $\sum_i \alpha_i \mu_i =\mu$,
\end{enumerate}
\end{theorem}
\begin{proof}
The proofs for all bounds can be found in Appendix C.
\end{proof}

 {The next two theorems generalize Theorems~\ref{theorem_bound_1} and \ref{theorem_bound_biased_1}. The proofs can be found in Appendix C.}

\begin{theorem}
\label{theorem_generalized_bound_1}
For any $t$, if all the $n$ sampling techniques are unbiased, the variance of $F^1$ is upper bounded as
\begin{equation}
\label{generalized_bound}
 V[F^1] \leq \frac{\left({\cal{H}} ({v_i}^t;\alpha_i)\right)^2}{{\cal{H}} ({v_i}^{2t-1};\alpha_i)}  + 
\mu^2 \left( \frac{\left({\cal{H}} ({v_i}^t;\alpha_i)\right)^2}{{\cal{H}}(v_i^{2t};\alpha_i)} -1 \right).
\end{equation}
\end{theorem}
\begin{theorem}
\label{theorem_generalized_bound_biased_1}
For any $t$, the variance of $F^1$ is upper bounded as
\begin{equation}
\label{generalized_bound_biased}
 V[F^1] \leq \frac{\left({\cal{H}} ({v_i}^t;\alpha_i)\right)^2}{{\cal{H}} ({v_i}^{2t-1};\alpha_i)}  + 
\frac{\left({\cal{H}} ({v_i}^t;\alpha_i)\right)^2}{{\cal{H}}(\frac{v_i^{2t}}{\mu_i^2};\alpha_i)} -\mu^2.
\end{equation}
\end{theorem}

Observe that the three cases in Theorems~\ref{theorem_bound_1}, and \ref{theorem_bound_biased_1} correspond to $t=0, t=1$, and $t=1/2$, respectively.
\begin{nota}
Considering that the arithmetic mean is the inverse of harmonic mean of inverse values, and after changing $-t$ by $t$, the bound in Theorem~\ref{theorem_generalized_bound_1} can be written too as
\begin{equation}
\label{generalized_bound2}
{V[F^1] \leq}  \frac{{\cal{A}} ({v_i}^{2t+1};\alpha_i)}{\left({\cal{A}} ({v_i}^t;\alpha_i)\right)^2}  + 
\mu^2 \left( \frac{{\cal{A}}(v_i^{2t};\alpha_i) }{\left({\cal{A}} ({v_i}^t;\alpha_i)\right)^2} -1 \right).
\end{equation}
And for biased techniques
\begin{equation}
\label{generalized_bound_biased2}
{V[F^1] \leq}  \frac{{\cal{A}} ({v_i}^{2t+1};\alpha_i)}{\left({\cal{A}} ({v_i}^t;\alpha_i)\right)^2}  + 
\frac{{\cal{A}}(\mu_i^2 v_i^{2t};\alpha_i) }{\left({\cal{A}} ({v_i}^t;\alpha_i)\right)^2} -\mu^2.
\end{equation}
Observe that the three cases in Theorems~\ref{theorem_bound_1},\ref{theorem_bound_biased_1} correspond now to $t=0, t=-1$, and $t=-1/2$, respectively.
\end{nota}

 \subsection{Case 2: fixed $\{\alpha_i\}_{i=1}^n$}
Consider now that $\{\alpha_i\}_{i=1}^n$ are fixed, and hence also $\{ {\sigma'}_i^2 \}_{i=1}^n$, are fixed. From Cauchy-Schwartz inequality,
 \begin{eqnarray}
\label{optimalbeta0}
\left(\sum_{i=1}^n \alpha_i \sigma'_i\right)^2 \le \left(\sum_{i=1}^n \beta_i \right)
\left(\sum_{i=1}^n \frac {\alpha_i^2 {\sigma'}_i^2}{\beta_i}\right),
\end{eqnarray}
Equality can only happen when for all $i$, $\beta_i \propto \frac {\alpha_i ^2{\sigma'}_i^2}{\beta_i}$,
thus the optimal $\{\beta_i\}_{i=1}^n$ are given by 
 \begin{eqnarray}
\label{optimalbeta}
\beta_i^* \propto \alpha_i \sigma'_i,\quad i=1,...n,
\end{eqnarray}
and the optimal (minimum) variance is
 \begin{eqnarray}
\label{optimalbeta2}
V [G^{1*}] = \left(\sum_{i=1}^n \alpha_i \sigma'_i\right)^2.
\end{eqnarray}
\begin{theorem}
Given an estimator $F$ with $\{\alpha_i\}_{i=1}^n$ values, we can always find a better estimator $G$ by sampling as $\beta_i^* \propto \alpha_i \sigma'_i$, which is strictly better whenever not all ${\sigma'}_i^2$ are equal.
\end{theorem}
\begin{proof}
Observe that, by Cauchy-Schwartz inequality,
 \begin{eqnarray}
\label{optimalbeta3}
\left(\sum_{i=1}^n \alpha_i \sigma'_i \right)^2 \le \left(\sum_{i=1}^n \alpha_i \right)
\left(\sum_{i=1}^n \alpha_i {\sigma'}_i^2 \right),
\end{eqnarray}
and hence, for the optimal values $\{\beta_i^*\}_{i=1}^n$ as in Eq. \eqref{optimalbeta}, the estimator $G^*$ always outperforms the estimator $F$ (in Eq. \eqref{optimalbeta3}, the left hand side is $V [G^1] $ while the right hand side is $V [F^1] $, since $\sum_{i=1}^n \alpha_i = 1$. Equality in Eq.~\eqref{optimalbeta3} only happens when for all $i$, $\alpha_i \propto \alpha_i {\sigma'}_i^2$, i.e., when all ${\sigma'}_i^2$ are equal.
\end{proof}
\begin{nota}
Note that, observing the two members {on the right-hand side} of Eq.~\eqref{optimalbeta3}, the maximum possible acceleration by using the optimal $\beta_i^*$ values when for all $i$, $\alpha_i = 1/n$ is equal to $n$ (\cite{SbertHavranSzirmay2016}).
\end{nota}

Let us now take into account the cost of each sampling technique is different, as it is usually considered in the literature (\cite{rubinstein2008simulation}). Let us denote the cost of sampling technique $i$ as  $c_i$. The inverse of efficiency for the estimator $G$ is given by 
\begin{eqnarray}
\label{optimalbeta4}
E_G^{-1} = \left(\sum_{i=1}^n \beta_i c_i \right) \left(\sum_{i=1}^n \frac {\alpha_i^2 {\sigma'}_i^2}{\beta_i} \right). 
\end{eqnarray}
Note that this quantity represents the total cost multiplied by the variance of the estimator. Using Cauchy-Schwartz inequality,
\begin{eqnarray}
\label{optimalbeta5}
\left(\sum_{i=1}^n \alpha_i \sigma'_i \sqrt{c_i} \right)^2 \le \left(\sum_{i=1}^n \beta_i c_i \right)
\left(\sum_{i=1}^n \frac {\alpha_i^2 {\sigma'}_i^2}{\beta_i} \right).
\end{eqnarray}
The optimal sampling rates (for maximizing the efficiency) are those that yield Eq. \eqref{optimalbeta5} as an equality, which happens when $\beta_i^* \propto \frac{\alpha_i {\sigma'}_i }{\sqrt{c_i}}$. Observe that, using again the Cauchy-Schwartz theorem, 
\begin{eqnarray}
\label{optimalbeta6}
\left(\sum_{i=1}^n \alpha_i \sigma'_i \sqrt{c_i} \right)^2 \le \left(\sum_{i=1}^n \alpha_i c_i \right) \left(\sum_{i=1}^n \alpha_i {\sigma'}_i^2 \right),
\end{eqnarray}
where the left hand side is $E_G^{-1}$ with the optimal sampling rates, and the right hand side is $E_F^{-1}$. Note that equality only happens when for all $i$, $c_i \propto {\sigma'}_i^2$. This is summarized in the following theorem.
\begin{theorem}
Given an estimator $F$ with $\{\alpha_i\}_{i=1}^n$ values, and sampling costs $\{c_i\}_{i=1}^n$, we can always find a more efficient estimator $G^*$ when for all $i$,  $\beta_i^* \propto \alpha_i \frac{\sigma'_i}{\sqrt{c_i}}$, which is strictly more efficient whenever not all $c_i \propto {\sigma'}_i^2$.
\end{theorem}

A particular case is when $\alpha_i = \frac{1}{n}$, $\forall i$, then
the variance becomes
  \begin{eqnarray}
\label{anewfamily6bis2}
V [G^1] = \frac{1}{n^2 }  \sum_{i=1}^n \frac{{\sigma'}_i^2 }{\beta_i}.
 \end{eqnarray}
This case was introduced in \cite[Section 4]{SbertHavranSzirmay2016}. It was shown that this estimator is provably better than $F$ with $\alpha_i = 1/n, \forall i$ when
 \begin{eqnarray}
\label{anewfamily6bis3}
\beta_i^* \propto \sigma'_i,\quad i=1,...,n,
\end{eqnarray}
which is the optimal case of Eq.~\eqref{anewfamily6bis2}. Examples showing the improvement obtained were also given in~\cite{SbertHavranSzirmay2016}.

 \subsection{Case 3: fixed $\{\beta_i\}_{i=1}^n$}
 \begin{theorem}
  Consider now a fixed set $\{\beta_i\}_{i=1}^n$. The optimal set $\{\alpha_i^*\}_{i=1}^n$ can be found using Lagrange multipliers with target function $$\Lambda (\{\alpha_i\}_{i=1}^n, \lambda)= \sum_{i=1}^n \frac{\alpha_i^2 {\sigma'}_i^2}{\beta_i} + \lambda \left(\sum_{i=1}^n \alpha_i - 1 \right).$$
 Observe that the ${\sigma'}_i^2$ values depend on the $\{\alpha_i\}_{i=1}^n$ values. The optimal values are those that obey, for all $j$, the following expression
 \begin{align}
\label{betafixed1}
& \frac{\alpha_j^* {\sigma'}_j^2}{\beta_j}  =  \sum_{i=1}^n  \frac{\alpha_i^{*2} }{\beta_i}  \times \nonumber \\ 
& \left( \int \frac{f^2(\x) p_i(\x) p_j(\x)}{(\sum_{k=1}^n \alpha_k^* p_k(\x))^3} d\x
 - {\mu'_i} \int \frac{f(\x) p_i(\x) p_j(\x)}{(\sum_{k=1}^n \alpha_k^* p_k(\x))^2}d\x \right).
\end{align}
\end{theorem}
\begin{proof}
The derivation can be found in the Appendix D.
\end{proof}
\begin{nota}
Note that in the general case, the optimal $\alpha_i^* \ne \beta_i$. However, a particular case when $\alpha_i^* = \beta_i$, is when all values $\mu'_i$ happen to be equal for these $\alpha_i^*$ values, and thus Eq.~\eqref{betafixed1} is satisfied for all $j$. See Appendix D for a further explanation. This is in concordance with Theorems 2 and 4 in \cite{Veach:PhD}. 
\end{nota}

\subsection{Case 4: $\beta_i = 1/n, \forall i$}
\label{allbetaequal}
In the case when for all $i$, $\beta_i = 1/n$, the variance becomes
 \begin{eqnarray}
\label{anewfamily6bis4}
V [G^1] = \sum_{i=1}^n \frac{ \alpha_i^2 {\sigma'}_i^2 }{1/n} = n \sum_{i=1}^n \alpha_i^2 {\sigma'}_i^2.
 \end{eqnarray}
Note that this is a usual case in the MIS literature strategies (\cite{ElviraMLB15,elvira2019generalized,elvira2016multiple,elvira2016heretical}) and the adaptive IS (AIS) literature (\cite{Cappe04,CORNUET12,martino2017layered,elvira2017improving,bugallo2017adaptive}), since all the techniques have the same number of counts.
By setting in Eq. \eqref{betafixed1} for all $i$, $\beta_i =1/n$, and if we can optimize $\{ \alpha_j\}_{j=1}^n$, we can find the minimum variance values $\{ \alpha_j^*\}_{j=1}^n$. 
% By setting in Eq. \eqref{betafixed1} for all $i$, $\beta_i =1/n$, and if we can optimize $\{ \alpha_j\}_{j=1}^n$, 
Thus the minimum variance $V[G^*]$ corresponds to the values $\{ \alpha_j^*\}_{j=1}^n$ that satisfy
 \begin{align}
\label{optimallagrange8} 
& \alpha_j^* {\sigma'}_j^2  =  \sum_{i=1}^n  \alpha_i^{*2} \\ \nonumber
& \times  \left( \int \frac{f^2(\x) p_i(\x) p_j(\x)}{(\sum_{k=1}^n \alpha_k^* p_k(\x))^3} d\x
 - {\mu'_i} \int \frac{f(\x) p_i(\x) p_j(\x)}{(\sum_{k=1}^n\alpha_k^* p_k(\x))^2}d\x \right).
\end{align}
The corresponding variance $V[G^*]$ will be less or equal than the variance of $V[G]$ for all $\{ \alpha_j\}_{j=1}^n$, and in particular for $\alpha_i =1/n$, where $G$ converts into $F$ as $\beta_i = \alpha_i =1/n$, the classic balance heuristic estimator, thus $V[G^*] \le V[F]$.

Apart from the optimal value $\{ \alpha_j^*\}_{j=1}^n$, we can find cases where $V[G^1] \le V[F^1]$ for $\beta_i =1/n$, 
 \begin{theorem}
 \label{theorem_LD1}
 If for $i < j$, $\alpha_i {\sigma'}_i^2 \le \alpha_j {\sigma'}_j^2 \Rightarrow \alpha_i \ge \alpha_j $, then
    \begin{eqnarray}
\label{anewfamily12} 
 V[F^1] \ge V[G^1]
  \end{eqnarray}
 \end{theorem}
 \begin{proof}
 We can write the inequality  \eqref{anewfamily12} as
     \begin{eqnarray}
\label{anewfamily13} 
 \sum_{i=1}^n \alpha_i {\sigma'}_i^2  &\ge& n \sum_{i=1}^n \alpha_i^2 {\sigma'}_i^2 \\
 \sum_{i=1}^n \frac{1}{n} (\alpha_i {\sigma'}_i^2) &\ge&  \sum_{i=1}^n \alpha_i (\alpha_i {\sigma'}_i^2),
  \end{eqnarray}
 and, as in the case of the hypothesis of the theorem, there is likelihood-dominance (\cite{Belzunce}) of sequence $1/n$ over sequence $\alpha_i$, then  Eq. \eqref{anewfamily13} holds. See also \cite{SbertPoch2016}, \cite{SbertHavranSzirmayElvira2018}.
 \end{proof}
 Observe that equality in Eq.~\eqref{anewfamily12} happens when $\alpha_i \propto 1/ {\sigma'}_i^2$.\\
  It can be equally proved the reverse case of Theorem~\ref{theorem_LD1} in the following theorem.
  \begin{theorem}
  \label{theorem_LD2}
  If for $i < j$, $\alpha_i {\sigma'}_i^2 \le \alpha_j {\sigma'}_j^2 \Rightarrow \alpha_i \le \alpha_j $, then
   \begin{eqnarray}
 \label{anewfamily12bis} 
  V[F^1] \le V[G^1].
   \end{eqnarray}
%   \victor{Es esto correcto?}
  \end{theorem}
  \begin{proof}
  {In that case where there is likelihood-dominance \citep{Belzunce} of sequence $\alpha_i$ over sequence $1/n$, then} %Eq.~\eqref{anewfamily14} holds,
 \begin{eqnarray}
 \label{anewfamily14}
  \big( \sum_{i=1}^n \alpha_i {\sigma'}_i^2 \big) \le n \big( \sum_{i=1}^n \alpha_i^2 {\sigma'}_i^2 \big).
  \end{eqnarray}
  \end{proof}

%%%%%%%%%%%%%%%%%%%%%%%%%

\section{Numerical examples}
\label{1D examples}
\subsection{Efficiency comparison between $F$ and $G$ estimators}
\label{comparison}
We compare the efficiencies for the $F$ estimator and the optimal $G$ estimator in 5 different examples. Table \ref{table:functions1} shows the inverse of the efficiencies, $E_F^{-1}=V[F]\cdot Cost[F]$ and $E_G^{-1}=V[G]\cdot Cost[G]$, i.e., the product of variance and cost for the $F$ estimator and for the optimal $G$ estimator, for these possible sets of $\{ \alpha_k \}_{k=1}^n$: (i) equal count of samples, (ii) count inversely proportional to variances of independent techniques (\cite{vrcaiHavranSbert14}), (iii) the new heuristic defined in
Section 6 of \cite{Sbert2018}, (iv) optimal count in \cite{SbertHavranSzirmay2016}, (v) and the two balance heuristic provably better estimators defined in \cite[Sections 4 and 5]{Sbert2018}. In the following, we describe the 5 examples.

%--------------
\subsubsection*{Example 1}
Suppose we want to solve the integral
\begin{equation}
\mu = \int_{\frac{3}{2\pi}}^{\pi} \x \left(\x^2-\frac{\x}{\pi }\right) \sin(\x) {\mathrm d}\x \approx 10.29
\end{equation}
 by MIS sampling on functions $\x$, $(\x^2-\frac{\x}{\pi })$, and
 $\sin(\x)$, respectively.  We first find the normalization constants:
 $\int_{\frac{3}{2\pi}}^{\pi} \x {\mathrm d}\x = 4.82$, $\int_{\frac{3}{2\pi}}^{\pi} (\x^2-\frac{\x}{\pi })
 {\mathrm d}\x = 8.76 $, $\int_{\frac{3}{2\pi}}^{\pi} \sin(\x) {\mathrm d}\x = 1.89$. The costs for sampling the techniques are
(1; 6.24; 3.28).
\vspace*{2mm}

%--------------
\subsubsection*{Example 2}
Let us solve the integral
\begin{equation}
\mu = \int_{\frac{3}{2\pi}}^{\pi}  \left(\x^2-\frac{\x}{\pi }\right) \sin^2 (\x) {\mathrm d}\x  \approx 3.60
\end{equation}
using the same functions $\x$, $(\x^2-\frac{\x}{\pi })$, and $\sin(\x)$ as
before.
%

%--------------
\subsubsection*{Example 3}
As the third example, let us solve the integral
\begin{equation}
\mu = \int_{\frac{3}{2\pi}}^{\pi}  \x+\left(\x^2-\frac{\x}{\pi }\right) + \sin (\x) {\mathrm d}\x \approx  15.47\hspace*{3mm}
\end{equation}
using the same functions as before.

%--------------
\subsubsection*{Example 4}
As the last example, consider the integral of the sum of the three
pdfs
\begin{align}
\hspace*{-3mm}\mu &= \int_{\frac{3}{2\pi}}^{\pi} 30 \frac{\x}{4.82082}+ 30 \frac{\left(x^2-\frac{\x}{\pi }\right)}{8.76463} + 40 \frac{\sin (\x)}{1.88816} {\mathrm d}\x \\ \nonumber
& \approx  100.
\end{align}
In this case we know the optimal (zero variance) $\alpha$ values:
$(0.3,0.3,0.4)$. This case should be most favorable to equal count of
samples.
\subsubsection*{Example 5}
As the last example, consider solving the following integral 
\begin{eqnarray}
\hspace*{-3mm}\mu  =  \int_{0.01}^{\pi/2} \left(\sqrt{\x} + \sin{\x} \right){\mathrm d}\x \approx 2.31175.
\end{eqnarray}
by MIS sampling on functions $2-\x$, and $\sin^2(\x)$.

\subsection{Bounds for the variance of $F$ estimator}
\label{bounds}

In Tables~\ref{table:bounds_example1}, \ref{table:bounds_example2}, \ref{table:bounds_example3}, and \ref{table:bounds_example4}, we give the bounds for Examples 1-4 of Section~\ref{comparison}, respectively. We use the same set of values $\{ \alpha_k \}_{k=1}^n$ in column 1, but now we consider equal cost of sampling. The last column contains the real  variances, approximated numerically with high precision. The second, fourth and fifth column contain the upper bounds. B1 is the upper bound based on the weighted harmonic mean of Eq.~\eqref{bound21}. B2 is the upper bound based on the weighted arithmetic mean in Eq.~\eqref{bound11}. Finally, B3 is the upper bound based on the weighted power mean in Eq.~\eqref{bound31}. For the sake of comparison we have included in column 2 and 4 the corresponding means of $\{v_i\}$ weighted with the $\{\alpha_i\}$ (B2 is also the arithmetic mean). From these four tables we can extract the following conclusions:
\begin{itemize}
\item {As expected, the listed bound values are indeed upper bounds for the variances.}
\item {None of the bounds is always the tightest.}
\item{B3 is always the tightest bound except in one case (second row of Example 3)}
\item{In examples 3 and 4, the bounds are much less tight, while in Example 2 the bounds are very tight.}
\item{The weighted harmonic mean is in most cases tighter than the bounds. However, we recall that it is not a bound, as we observe in Example 4.}
\end{itemize}

\begin{table*}
\caption{\label{table:bounds_example1} Upper bounds for the variances of $F$ estimator for Example 1 
%with equal costs (1,1,1).
}
\begin{center}
  {\footnotesize
\begin{tabular}{ |c|c|c|c|c|c|c|}
 \hline
&B1&harmonic mean&B2 (arithmetic mean) & B3&power mean (-1/2)&variance\\
\hline
\hline
 $\alpha_k \propto \frac{1}{n}$ &59,8863&33,6961&53,7493&46,4125&36,767&29,1634\\
$\alpha_k \propto \frac{1}{v_k}$&34,2727&27,0116&33,6961&30,876&27,7974&24,1116\\
 $\alpha_k \propto \frac{1}{m^2_k}$&47,7525&30,4426&45,0066&39,271&32,4148&26,5536\\
$\alpha_k \propto \sigma_{k,eq}$ &62,3836&34,4548&55,4324&47,8705&37,7497&29,0908\\
$\alpha_k \propto M_{k,eq}$&56,2199&32,6544&51,2191&44,2741&35,3942&28,2435\\
 \hline
\end{tabular}
}
\end{center}
\end{table*}

\begin{table*}
\caption{\label{table:bounds_example2} Upper bounds for the variances of $F$ estimator for Example 2 
%with equal costs (1,1,1).
}
\begin{center}
  {\footnotesize
\begin{tabular}{ |c|c|c|c|c|c|c|}
 \hline
&B1&harmonic mean&B2 (arithmetic mean) & B3&power mean (-1/2)&variance\\
\hline
\hline
 $\alpha_k \propto \frac{1}{n}$ &6,96851&5,9558&6,53264&6,36347&6,08435&4,9176\\
$\alpha_k \propto \frac{1}{v_k}$&6,25335&5,52328&5,9558&5,82562&5,61376&4,5528\\
 $\alpha_k \propto \frac{1}{m^2_k}$&6,70447&5,79234&6,32076&6,16385&5,90726&4,7754\\
$\alpha_k \propto \sigma_{k,eq}$ &7,05992&6,02603&6,61134&6,44027&6,1577&4,9992\\
$\alpha_k \propto M_{k,eq}$&6,85368&5,87468&6,43664&6,27103&5,99849&4,8324\\
 \hline
\end{tabular}
}
\end{center}
\end{table*}

\begin{table*}
\caption{\label{table:bounds_example3} Upper bounds for the variances of $F$ estimator for Example 3 
%with equal costs (1,1,1).
}
\begin{center}
  {\footnotesize
\begin{tabular}{ |c|c|c|c|c|c|c|}
 \hline
&B1&harmonic mean&B2 (arithmetic mean) & B3&power mean (-1/2)&variance\\
\hline
\hline
 $\alpha_k \propto \frac{1}{n}$ &355,59&11,0158&3208,72&213,213&19,9094&10,6877\\
$\alpha_k \propto \frac{1}{v_k}$&25,7535&4,51631&11,0158&15,9986&4,72888&2,02066\\
 $\alpha_k \propto \frac{1}{m^2_k}$&148,369&7,10414&142,853&66,8855&8,82679&0,368009\\
$\alpha_k \propto \sigma_{k,eq}$ &308,038&10,0157&2942,08&189,329&17,2153&9,48229\\
$\alpha_k \propto M_{k,eq}$&330,312&10,6197&2769,89&188,646&18,1804&7,05337\\
 \hline
\end{tabular}
}
\end{center}
\end{table*}

\begin{table*}
\caption{\label{table:bounds_example4} Upper bounds for the variances of $F$ estimator for Example 4 
%with equal costs (1,1,1).
}
\begin{center}
  {\footnotesize
\begin{tabular}{ |c|c|c|c|c|c|c|}
 \hline
&B1&harmonic mean&B2 (arithmetic mean) & B3&power mean (-1/2)&variance\\
\hline
\hline
 $\alpha_k \propto \frac{1}{n}$ &18463,3&814,05&57587,8&11878,3&1646,94&28,1431\\
$\alpha_k \propto \frac{1}{v_k}$&685,56&294,421&814,05&573,436&302,401&330,852\\
 $\alpha_k \propto \frac{1}{m^2_k}$&6335,38&459,805&8518,29&4107,74&620,166&809,287\\
$\alpha_k \propto \sigma_{k,eq}$ &15780,7&733,625&52772,4&10455,1&1398,07&46,618\\
$\alpha_k \propto M_{k,eq}$&18681,5&819,662&59442,8&12108,1&1674,84&18,0465\\
 \hline
\end{tabular}
}
\end{center}
\end{table*}

%%%%%%%%%%%%%%%%%%%%%%%%%

\section{Conclusions}
\label{sec:conclusion}
%%%%%%%%%%%%%%%%%%%%%%%%%
In this paper, we have proposed a multiple importance sampling estimator that combines samples simulated from different techniques. The novel estimator generalizes the balance heuristic estimator, widely used Monte Carlo in the literature of signal processing, computational statistics, and computer graphics. In particular, this estimator relaxes the connection between the coefficients that select the number of samples per proposal, and the samples that appear in the mixture of techniques at the denominator of the importance weight. This flexibility shows a relevant improvement in terms of variance in the combined estimator w.r.t. the balance heuristic estimator (which is include as a particular case in the novel estimator). We have studied the optimal choice of the free coefficients in such a way the variance of the resulting estimator is minimized. In addition, numerical results have shown that the significant gap in terms of variance between both estimators justifies the use of the novel estimator whenever possible. We have also presented novel bounds for the variance of the balance heuristic estimator.

\appendix
%\section{Appendices}

%%%%%%%%%%%%%%%%%%%%%%%%%%%%%%%%%%%
\section*{Appendix A: difference between the variances of deterministic and randomized multiple importance sampling estimators}
The difference between the variances of the deterministic multiple importance sampling estimator, $Z$, and the randomized one, ${\cal Z}$, is given by \cite{Veach:PhD} (we normalize here  to one sample)
\begin{eqnarray}
\label{upperbounddifference1}
V[{\cal Z}^1] - V[Z^1] &=& \sum_i \alpha_i {\mu'}_i^2 - \mu^2 \\ \nonumber
&=&  \sum_i \alpha_i {\mu'}_i^2 - (\sum_i \alpha_i {\mu'}_i)^2
\end{eqnarray}
As by Cauchy-Schwartz
\begin{eqnarray}
\label{upperbounddifference2}
(\sum_i \alpha_i {\mu'}_i)^2 \le (\sum_i \alpha_i) (\sum_i \alpha_i {\mu'}_i^2),
\end{eqnarray}
equality only happens (apart from the case when both variances $V[Z^1],V[{\cal Z}^1]$ are zero) when for all $i$, $\alpha_i \propto\alpha_i {\mu'}_i^2$, i.e., when all ${\mu'}_i$ are equal. Observe that we can write
\begin{eqnarray}
\label{upperbounddifference3}
{\cal A}^2 ({\mu'}_i; \alpha_i) \le {\cal PWM}(2)^2 ({\mu'}_i; \alpha_i) 
\end{eqnarray}
where ${\cal A} ({\mu'}_i; \alpha_i)$ is the weighted arithmetic mean of $\{{\mu'}_i\}$ values with weights $\{\alpha_i\}$, and ${\cal PWM}(2)$ is the power mean with power 2 (observe that arithmetic mean is the power mean with power 1).
When for all $i$, $\alpha_i = 1/n$,  
\begin{eqnarray}
\label{upperbounddifference4}
 \mu^2 = {\cal PWM}(2)^2 ({\mu'}_i; 1/n) \le n {\cal A}^2 ({\mu'}_i; 1/n) = n \mu^2.
\end{eqnarray}
Thus when $\alpha_i = 1/n$,
\begin{eqnarray}
\label{upperbounddifference5}
V[{\cal Z}^1] - V[Z^1] \le (n-1) \mu^2. 
\end{eqnarray}

\section*{Appendix B. Proof of Theorem \ref{theorem_variance_F}: Optimal variance of $F$}
 The $\{\alpha_i\}_{i=1}^n$ values for the optimal variance of $F$ estimator can be obtained using Lagrange multipliers with the target function $$\Lambda (\{\alpha_i\}_{i=1}^n, \lambda)= \sum_{i=1}^n \alpha_i {\sigma'}_i^2 + \lambda \left(\sum_{i=1}^n \alpha_i =1\right).$$
 Taking partial derivatives with respect to {$\alpha_j$},
  \begin{eqnarray}
 \label{optimallagrangeF1}
 \frac {\partial \Lambda (\{\alpha_i\}_{i=1}^n, \lambda)} {\partial_{\alpha_j}} &=&  \frac {\partial \left( \sum_{i=1}^n \alpha_i {\sigma'}_i^2 \right)} {\partial_{\alpha_j}}     +  \frac {\partial  \left( \lambda \left(\sum_{i=1}^{n} \alpha_i - 1 \right) \right)} {\partial_{\alpha_j}} \\ \nonumber
 &=&   \sum_{i=1}^n \frac  {\partial \left( \alpha_i {\sigma'}_i^2 \right)} {\partial_{\alpha_j}}        +  \lambda   = 0.
 \end{eqnarray}
 The partial derivatives are equal to
  \begin{eqnarray}
 \label{optimallagrangeF2}
  \frac  {\partial \left( \alpha_i {\sigma'}_i^2 \right)} {\partial_{\alpha_j}} =  \delta_{ij} {\sigma'}_j^2 +\alpha_i \frac  {\partial \left( {\sigma'}_i^2 \right)} {\partial_{\alpha_j}}.
 \end{eqnarray}
where $\delta_{ij}$ is Dirac's delta function, and
 \begin{eqnarray}
\label{betafixedproof3}
\frac  {\partial \left( {\sigma'}_i^2 \right)} {\partial_{\alpha_j}} &=& \frac  {\partial \left( \int \frac{f^2(\x) p_i(\x)}{(\sum_{k=1}^n \alpha_k p_k(\x))^2} d\x -  ( {\mu'_i})^2 \right)} {\partial_{\alpha_j}} \\ \nonumber
&=&  \frac  {\partial \left( \int \frac{f^2(\x) p_i(\x)}{(\sum_{k=1}^n \alpha_k p_k(\x))^2} d\x \right)} {\partial_{\alpha_j}} - 2 {\mu'_i} \frac  {\partial {\mu'_i}}{\partial_{\alpha_j}} \\ \nonumber
&=&  -2 \int \frac{f^2(\x) p_i(\x) p_j(\x)}{(\sum_{k=1}^n \alpha_k p_k(\x))^3} d\x - 2 {\mu'_i} \frac  {\partial {\mu'_i}}{\partial_{\alpha_j}}.
\end{eqnarray}
Since we can write
 \begin{eqnarray}
\label{betafixedproof4}
\frac  {\partial {\mu'_i}}{\partial_{\alpha_j}} &=& \frac  {\partial \left( \int \frac{f(\x) p_i(\x)}{\sum_{k=1}^n \alpha_k p_k(\x)}d\x \right)}{\partial_{\alpha_j}} \\ \nonumber
&=& - \int \frac{f(\x) p_i(\x) p_j(\x)}{(\sum_{k=1}^n \alpha_k p_k(\x))^2}d\x,
\end{eqnarray}
% \cblue{Using the result in Eq.~\eqref{betafixedproof5}, we obtain,}
thus Eq.~\eqref{betafixedproof3} reads
 \begin{eqnarray}
\label{betafixedproof5}
 \frac  {\partial \left({\sigma'}_i^2 \right)} {\partial_{\alpha_j}} &=& -2 \int \frac{f^2(\x) p_i(\x) p_j(\x)}{(\sum_{k=1}^n \alpha_k p_k(\x))^3} d\x \\ \nonumber
 &+& 2 {\mu'_i} \int \frac{f(\x) p_i(\x) p_j(\x)}{(\sum_{k=1}^n \alpha_k p_k(\x))^2}d\x.
\end{eqnarray}
Then,
  \begin{eqnarray}
 \label{optimallagrangeF3}
& &\lambda =  - \sum_{i=1}^n \frac  {\partial \left( \alpha_i {\sigma'}_i^2 \right)} {\partial_{\alpha_j}} =  -{\sigma'}_j^2  + 2  \sum_{i=1}^n  \alpha_i \times\\ \nonumber
 & &  \left( \int \frac{f^2(\x) p_i(\x) p_j(\x)}{(\sum_{k=1}^n \alpha_k p_k(\x))^3} d\x
  -  {\mu'_i} \int \frac{f(\x) p_i(\x) p_j(\x)}{(\sum_{k=1}^n \alpha_k p_k(\x))^2}d\x \right) \\ \nonumber
 &=& -{\sigma'}_j^2 + 2 \int \frac{f^2(\x) p_j(\x) (\sum_{i=1}^n \alpha_i p_i(\x)) }{(\sum_{k=1}^n \alpha_k p_k(\x))^3} d\x \\ \nonumber
  &-& 2 \sum_{i=1}^n \alpha_i {\mu'_i} \int \frac{f(\x) p_i(\x) p_j(\x)}{(\sum_{k=1}^n \alpha_k p_k(\x))^2}d\x \\ \nonumber
  &=& - {\sigma'}_j^2  +2 ( {\sigma'}_j^2 + {\mu'_j}^2)
 - 2 \sum_{i=1}^n \alpha_i {\mu'_i} \int \frac{f(\x) p_i(\x) p_j(\x)}{(\sum_{k=1}^n \alpha_k p_k(\x))^2}d\x \\ \nonumber
  &=& {\sigma'}_j^2 +2 {\mu'_j}^2
 - 2\sum_{i=1}^n \alpha_i {\mu'_i} \int \frac{f(\x) p_i(\x) p_j(\x)}{(\sum_{k=1}^n \alpha_k p_k(\x))^2}d\x \\ \nonumber
 \end{eqnarray}
 This is, for all $j$, the following values have to be equal,
   \begin{eqnarray}
 \label{optimallagrangeF3bis}
 {\sigma'}_j^2 +2 {\mu'_j}^2
 - 2\sum_{i=1}^n \alpha_i {\mu'_i} \int \frac{f(\x) p_i(\x) p_j(\x)}{(\sum_{k=1}^n \alpha_k p_k(\x))^2}d\x
 \end{eqnarray}
Multiplying by $\alpha_j$ and adding over all $j$ in Eq.~\eqref{optimallagrangeF3},
  \begin{eqnarray}
 \label{optimallagrangeF4}
  \lambda &=&  \sum_{j=1}^n \alpha_j \lambda = \sum_{j=1}^n \alpha_j   {\sigma'}_j^2 +2  \sum_{j=1}^n \alpha_j {\mu'_j}^2 \\ \nonumber
  &-& 2 \sum_{i=1}^n \alpha_i {\mu'_i} \int \frac{f(\x) p_i(\x) (\sum_{j=1}^n \alpha_j p_j(\x))}{(\sum_{k=1}^n \alpha_k p_k(\x))^2}d\x \\ \nonumber
  &=&  \sum_{j=1}^n \alpha_j   {\sigma'}_j^2 +2  \sum_{j=1}^n \alpha_j {\mu'_j}^2 - 2 \sum_{i=1}^n \alpha_i {\mu'_i}^2  \\ \nonumber
  &=& \sum_{j=1}^n \alpha_j   {\sigma'}_j^2
 \end{eqnarray} 
 which is the optimal variance of estimator {$F$}.  Observe that all derivatives in~\eqref{optimallagrangeF1} are negative for the optimal $\{\alpha_i^{\star}\}$ values and equal to $-\sum_{j=1}^n \alpha_j^{\star}   {\sigma^{'\star}}_j^2$.  Let us compare Eq.~\eqref{optimallagrangeF3bis} with the condition of minimum variance for estimator ${\cal{F}}$, the randomized version of $F$, which turns into equality for all $j$ of the sum~(\cite{SbertHavranSzirmay2016})
  \begin{eqnarray}
 \label{optimalrandomizedF1}
{\sigma'}_j^2 + {\mu'_j}^2. 
  \end{eqnarray} Observe that, when for given $\{\alpha_i\}$ for all $i$ the $\{\mu'_i\}$ values are equal (and thus equal to $\mu$), the Eq.~\eqref{optimallagrangeF3bis} and Eq.~\eqref{optimalrandomizedF1} become the same, for these $\{\alpha_i\}$ to be optimal all values  $\{{\sigma'}_i\}$ have to be equal too, and both minimum variances of $F$ and ${\cal{F}}$ are equal.

\section*{Appendix C: upper bounds for balance heuristic estimator}
The variance of {random mixture} generic MIS estimator  is given by
\begin{eqnarray}
\label{newbound1}
V[{\cal{Z}}] = \sum_{i=1}^{n} \int \frac{w_i^2(\x) f^2(\x)}{\alpha_i p_i(\x)} {\mathrm d}\x - \mu^2.
\end{eqnarray}
with optimal $w_i(\x)$ weights the balance heuristic weights (see \cite[Theorems 9.2,9.4]{Veach:PhD}). This is, the optimal case is for ${\cal{Z}} \equiv {\cal{F}}$. Let us consider now lineal combination of {the resulting estimators of the $n$ techniques $\{p_i(\x)\}_{i=1}^n$.} In this case, all weights $w_i(\x)\equiv w_i$ are constant in all domain. Let ${\cal{Z}}_{l1}$ be the randomized version of the estimator $Z_{l1}$, that is the optimal deterministic lineal combination for $\alpha_i$ fixed. The optimal weights are $w_i = \frac{H(v_i/\alpha_i) \alpha_i}{n v_i}$, and $V[Z_{l1}]$ is equal to $\frac{H(v_i/\alpha_i)}{n}$, where $v_i$ is the variance of technique $i$, see \cite{SbertHavran2017}. In order to find $V[{\cal{Z}}_{l1}]$, we substitute the optimal weights in Eq.~\eqref{newbound1} to obtain, {supposing all techniques are unbiased},
\begin{eqnarray}
\label{newbound2}
V[{\cal{Z}}_{l1}] &=&  \sum_{i=1}^{n} \int \frac{  (\frac{H(v_i/\alpha_i) \alpha_i}{n v_i}) ^2      f^2(\x)}{\alpha_i p_i(\x)} {\mathrm d}\x - \mu^2 \\ \nonumber
%\int \frac {( H(v_i/ \alpha_i) \alpha_i}{n v_i} )^2  f^2(\x)} {\alpha_i p_i(\x)} {\mathrm d}x - \mu^2 \\ \nonumber
&=& \frac{(H(v_i/ \alpha_i))^2}{n}  \sum_{i=1}^{n}  \frac{\alpha_i}{n v_i^2} \int  \frac{f^2(\x)} { p_i(\x)} {\mathrm d}\x - \mu^2 \\ \nonumber
&=& \frac{(H(v_i/ \alpha_i))^2}{n}  \left( \sum_{i=1}^{n}  \frac{\alpha_i}{n v_i^2} (v_i+\mu^2) \right) - \mu^2 \\ \nonumber
&=& \frac{(H(v_i/ \alpha_i))^2}{n}  \left( \sum_{i=1}^{n}  \frac{\alpha_i}{n v_i}+\mu^2 \sum_{i=1}^{n}  \frac{\alpha_i}{n v_i^2} \right)- \mu^2 \\ \nonumber
&=& \frac{(H(v_i/ \alpha_i))^2}{n}  \left( \frac{1}{H(v_i/ \alpha_i)} +\mu^2 \frac{1}{H(v_i^2/ \alpha_i)} \right)- \mu^2 \\ \nonumber
&=& \frac{H(v_i/\alpha_i)}{n} + \mu^2 \left( \frac{(H(v_i/\alpha_i))^2}{n H(v_i^2/\alpha_i)} -1 \right).
\end{eqnarray}
Thus by Theorem 9.4 in \cite{Veach:PhD}, $V[{\cal{Z}}_{l1}]$ in Eq.~\eqref{newbound2} is an upper bound for the variance of the random balance heuristic mixture estimator $V[\cal{F}]$, which in its turn is an upper bound for the deterministic balance heuristic mixture estimator, $V[{F}]$. The bound can be written in terms of weighted harmonic mean
\begin{eqnarray}
\label{newbound3}
V[{\cal{Z}}_{l1}]= {\cal{H}} (v_i;\alpha_i)  + \mu^2 \left( \frac{({\cal{H}}(v_i;\alpha_i))^2}{{\cal{H}}(v_i^2;\alpha_i)} -1 \right),
\end{eqnarray}
where ${\cal{H}} (v_i;\alpha_i)$ is the weighted harmonic mean of $\{v_i\}$ with weights $\{\alpha_i\}$.  Compare with the bound ${\cal{A}} (v_i;\alpha_i)$, the weighted arithmetic mean, obtained in \cite{SbertHavran2017}.\\
In case the techniques are biased, the bound would be, continuing from second line in Eq.~\eqref{newbound2},
\begin{eqnarray}
\label{newbound4}
V[{\cal{Z}}_{l1}] &=& \frac{(H(v_i/ \alpha_i))^2}{n}  \left( \sum_{i=1}^{n}  \frac{\alpha_i}{n v_i^2} (v_i+\mu_i^2) \right) - \mu^2 \\ \nonumber
&=& \frac{(H(v_i/ \alpha_i))^2}{n}  \left( \sum_{i=1}^{n}  \frac{\alpha_i}{n v_i}+\sum_{i=1}^{n}  \frac{\alpha_i \mu_i^2 }{n v_i^2} \right)- \mu^2 \\ \nonumber
&=& \frac{(H(v_i/ \alpha_i))^2}{n}  \left( \frac{1}{H(v_i/ \alpha_i)} + \frac{1}{H(\frac{v_i^2}{ \mu_i^2} /\alpha_i)} \right)- \mu^2 \\ \nonumber
&=& \frac{H(v_i/\alpha_i)}{n} +  \frac{(H(v_i/\alpha_i))^2}{n H(\frac{v_i^2}{ \mu_i^2} /\alpha_i)} -\mu^2 \\ \nonumber
&=& {\cal{H}} (v_i;\alpha_i)  + \frac{({\cal{H}}(v_i;\alpha_i))^2}{{\cal{H}}(\frac{v_i^2}{\mu_i^2};\alpha_i)} -\mu^2,
\end{eqnarray}
where $\mu_i$ is the expected value corresponding to technique $i$, and such that $\sum_i w_i \mu_i = \mu$.\\
Using other combinations of values for $w_i$ and $\alpha_i$ in the randomized linear combination of estimators we can obtain other bounds. The most interesting cases are when for all $i$, $w_i=\alpha_i$, which we call estimator $Z_{l2}$, and the optimal case when $w_i$ is fixed, which corresponds to $\alpha_i \propto w_i \sqrt{v_i}$ (\cite{SbertHavran2017}), which we call estimator $Z_{l3}$. In the first case, when for all $i$, $w_i=\alpha_i$, the variance of the randomized linear combination estimator ${\cal{Z}}_{l2}$ (and so an upper bound for the variance of balance heuristic estimators ${\cal{F}}$ and $F$), it is found to be
\begin{eqnarray}
\label{newbound5}
V[{\cal{Z}}_{l2}] &=& \sum_{i=1}^{n} \int \frac{\alpha_i f^2(\x)}{ p_i(\x)} {\mathrm d}\x - \mu^2 \\ \nonumber
&=& \sum_{i=1}^{n} \alpha_i \int \frac{ f^2(\x)}{ p_i(\x)} {\mathrm d}\x - \mu^2 \\ \nonumber
&=& \sum_{i=1}^{n} \alpha_i (v_i+ \mu^2) - \mu^2 \\ \nonumber
&=& \sum_{i=1}^{n} \alpha_i v_i \\ \nonumber
&=& {\cal{A}} (v_i;\alpha_i).
\end{eqnarray}
where ${\cal{A}} (v_i;\alpha_i)$ is the weighted arithmetic mean of $\{v_i\}$ with weights $\{\alpha_i\}$. This bound was already obtained in \cite{SbertHavran2017}. Observe also that in this case, the variances of randomized and deterministic  (\cite{SbertHavran2017}) linear combination are the same. This happens because being for all $i$, $w_i=\alpha_i$, all $\alpha_i {\mu'}_i$ values, with ${\mu'}_i = \frac{1}{\alpha_i} \int w_i(\x) f(\x) d\x$, Eq.~\eqref{differenceofZ2}, are the same, see Section~\ref{preliminaries} and Appendix A. For the case of biased techniques, the variance is
\begin{eqnarray}
\label{newbound5biased}
V[{\cal{Z}}_{l2}]= {\cal{A}} (v_i;\alpha_i)+ {\cal{A}} (\mu_i^2;\alpha_i)-\mu^2.
\end{eqnarray}\\
For the case $\alpha_i \propto w_i \sqrt{v_i}$, for the randomized estimator ${\cal{Z}}_{l3}$, we isolate $w_i$ so that 
$w_i = \frac{ \frac {\alpha_i} {\sqrt{v_i}}}{ \sum_k \frac {\alpha_k}{\sqrt{v_k}} } = \frac{ H(\frac {\sqrt{v_i}}{{\alpha_i}}) \frac {\alpha_i} {\sqrt{v_i}}}{ n} $ and substituting the $w_i$ values in Eq.~\eqref{newbound1} the variance is found to be
\begin{eqnarray}
\label{newbound6}
V[{\cal{Z}}_{l3}] &=&\left( \frac{H(\frac {\sqrt{v_i}}{{\alpha_i}})}{n}\right)^2 \sum_{i=1}^{n} \frac{ \alpha_i}{v_i } \int \frac{ f^2(\x)}{ p_i(\x)} {\mathrm d}\x - \mu^2\\ \nonumber
&=& \left( \frac{H(\frac {\sqrt{v_i}}{{\alpha_i}})}{n}\right)^2 \sum_{i=1}^{n} \frac{ \alpha_i}{v_i } (v_i+\mu^2) - \mu^2\\ \nonumber
&=& \left( \frac{H(\frac {\sqrt{v_i}}{{\alpha_i}})}{n}\right)^2 (1 + \frac{\mu^2 n} { H(\frac {v_i}{{\alpha_i}})}) - \mu^2\\ \nonumber
&=& \left({\cal{H}} (\sqrt{v_i};\alpha_i)\right)^2  + 
\mu^2 \left( \frac{\left({\cal{H}} (\sqrt{v_i};\alpha_i)\right)^2}{{\cal{H}}(v_i;\alpha_i)} -1 \right),
\end{eqnarray}
and for the case of biased techniques,
\begin{eqnarray}
\label{newbound7}
V[{\cal{Z}}_{l3}]= \left({\cal{H}} (\sqrt{v_i};\alpha_i)\right)^2  + 
\frac{\left({\cal{H}} (\sqrt{v_i};\alpha_i)\right)^2}{{\cal{H}}(\frac{v_i}{\mu_i^2};\alpha_i)} -\mu^2.
\end{eqnarray}
Observe that $\left({\cal{H}} (\sqrt{v_i};\alpha_i)\right)^2$ is the power mean, with power $-1/2$, of $\{v_i\}$ with weights $\{\alpha_i\}$. Remembering that harmonic mean is power mean with power $-1$, and arithmetic mean is power mean with power $1$, by the increasing property of the power mean function we can write,
\begin{eqnarray}
\label{newbound8}
{\cal{H}} (v_i;\alpha_i) \le \left({\cal{H}} (\sqrt{v_i};\alpha_i)\right)^2 \le {\cal{A}} (v_i;\alpha_i).
\end{eqnarray}
Although it is reasonable to expect that the three bounds in Eqs.~\eqref{newbound3}, \eqref{newbound6}, and \eqref{newbound5} follow the same order, we can not state this in general, as there might be cases where the second terms in Eqs.~\eqref{newbound3} and \eqref{newbound6} would make the bound result in a bigger bound than the one in Eq.~\eqref{newbound5}.

 Observe now that the three bounds  in Eqs.~\eqref{newbound3}, \eqref{newbound5}, and \eqref{newbound6} come from taking
 $w_i \propto \frac{\alpha_i}{v_i^t}$, with $t=1, t=0$, and $t=1/2$, respectively. By normalizing we have $w_i = {\cal{H}} (v_i^t;\alpha_i) \frac{\alpha_i}{v_i^t}$, and substituting into Eq.~\eqref{newbound1}  we obtain, for any $t$, when all techniques are unbiased the upper bound for  $V[F^1]$

\begin{equation}
\label{appendix_generalized_bound}
 \frac{\left({\cal{H}} ({v_i}^t;\alpha_i)\right)^2}{{\cal{H}} ({v_i}^{2t-1};\alpha_i)}  + 
\mu^2 \left( \frac{\left({\cal{H}} ({v_i}^t;\alpha_i)\right)^2}{{\cal{H}}(v_i^{2t};\alpha_i)} -1 \right).
\end{equation}
And for biased techniques $V[F^1]$ is upper bounded by
\begin{equation}
\label{appendix_generalized_bound_biased}
 \frac{\left({\cal{H}} ({v_i}^t;\alpha_i)\right)^2}{{\cal{H}} ({v_i}^{2t-1};\alpha_i)}  + 
\frac{\left({\cal{H}} ({v_i}^t;\alpha_i)\right)^2}{{\cal{H}}(\frac{v_i^{2t}}{\mu_i^2};\alpha_i)} -\mu^2.
\end{equation}

  \section*{Appendix D: Derivation of Case 3}
We present here the proof of Eq.~\eqref{betafixed1}. 
We have to optimize the target function $$\Lambda (\{\alpha_i\}_{i=1}^n, \lambda)= \sum_{i=1}^n \frac{\alpha_i^2 {\sigma'}_i^2}{\beta_i} + \lambda \left(\sum_{i=1}^n \alpha_i - 1\right).$$

Taking partial derivatives with respect to $\alpha_i$, as the $\beta_i$ values are constant,
 \begin{eqnarray}
\label{betafixedproof1}
\frac {\partial \Lambda (\{\alpha_i\}_{i=1}^n, \lambda)} {\partial_{\alpha_j}} 
=   \sum_{i=1}^n \frac{1}{\beta_i} \frac  {\partial \left( \alpha_i^2 {\sigma'}_i^2 \right)} {\partial_{\alpha_j}}        +  \lambda   = 0.
\end{eqnarray}
The partial derivatives are equal to
 \begin{eqnarray}
\label{betafixedproof2}
\frac  {\partial \left( \alpha_i^2 {\sigma'}_i^2 \right)} {\partial_{\alpha_j}} =  2 \alpha_j \delta_{ij} {\sigma'}_j^2 +\alpha_i^2 \frac  {\partial \left( {\sigma'}_i^2 \right)} {\partial_{\alpha_j}},
\end{eqnarray}
where $\delta_{ij}$ is Dirac's delta function. Using the result in Eq.~\eqref{betafixedproof5}, we obtain
 \begin{eqnarray}
\label{betafixedproof6}
& &   \sum_{i=1}^n \frac{1}{\beta_i} \frac  {\partial \left( \alpha_i^2 {\sigma'}_i^2 \right)} {\partial_{\alpha_j}} = 2 \frac{\alpha_j {\sigma'}_j^2}{\beta_j}  - 2  \sum_{i=1}^n  \frac{\alpha_i^2}{\beta_i} \\ \nonumber
&\times&  \left( \int \frac{f^2(\x) p_i(\x) p_j(\x)}{(\sum_{k=1}^n \alpha_k p_k(\x))^3} d\x
 - {\mu'_i} \int \frac{f(\x) p_i(\x) p_j(\x)}{(\sum_{k=1}^n \alpha_k p_k(\x))^2}d\x \right) \\ \nonumber
&=& -\lambda.
\end{eqnarray}
In Eq.~\eqref{betafixedproof6}, we multiply by $\alpha_j$, and add over all indexes $j$, obtaining
 \begin{eqnarray}
\label{betafixedproof7}
 \lambda &=& - 2 \sum_{j=1}^n \frac{\alpha_j^2 {\sigma'}_j^2}{\beta_j}   + 2  \sum_{i=1}^n  \frac{\alpha_i^2}{\beta_i}  \nonumber  \\
&\times&  \big( \int \frac{f^2(\x) p_i(\x) (\sum_{j=1}^n \alpha_j p_j(\x))}{(\sum_{k=1}^n \alpha_k p_k(\x))^3} d\x \nonumber \\ 
&-& {\mu'_i} \int \frac{f(\x) p_i(\x) (\sum_{j=1}^n \alpha_j p_j(\x))}{(\sum_{k=1}^n \alpha_k p_k(\x))^2}d\x \big)  \nonumber  \\
&=& - 2 \sum_{j=1}^n \frac{\alpha_j^2 {\sigma'}_j^2}{\beta_j}   + 2  \sum_{i=1}^n   \frac{\alpha_i^2}{\beta_i} \nonumber  \\ 
&\times&  \left( \int \frac{f^2(\x) p_i(\x) }{(\sum_{k=1}^n \alpha_k p_k(\x))^2} d\x
- {\mu'_i} \int \frac{f(\x) p_i(\x) }{(\sum_{k=1}^n \alpha_k p_k(\x))}d\x \right) \nonumber \\ 
&=& - 2 \sum_{j=1}^n \frac{\alpha_j^2 {\sigma'}_j^2}{\beta_j}   + 2  \sum_{i=1}^n  \frac{\alpha_i^2  {\sigma'}_i^2}{\beta_i} = 0.  
\end{eqnarray}
We remind that  $\sum_{j=1}^n \alpha_j = 1$ which disappears in the left-hand side and the second term of the right-hand side equation.

From Eq.~\eqref{betafixedproof6}, the optimal $\{ \alpha_j\}_{j=1}^n$ are those that obey
 \begin{eqnarray}
\label{betafixedproof8}
& & \frac{\alpha_j {\sigma'}_j^2}{\beta_j}   =  \sum_{i=1}^n  \frac{\alpha_i^2}{\beta_i}  \\ 
&\times&  \left( \int \frac{f^2(\x) p_i(\x) p_j(\x)}{(\sum_{k=1}^n \alpha_k p_k(\x))^3} d\x
 - {\mu'_i} \int \frac{f(\x) p_i(\x) p_j(\x)}{(\sum_{k=1}^n \alpha_k p_k(\x))^2}d\x \right).\nonumber
\end{eqnarray}
Observe that if we take for all $i$, $\alpha_i = \beta_i$ then
 \begin{eqnarray}
\label{betafixedproof9}
 {\sigma'}_j^2 &= &  \int \frac{f^2(\x) \left(\sum_{i=1}^n  \alpha_i p_i(\x) \right) p_j(\x)}{(\sum_{k=1}^n \alpha_k p_k(\x))^3} d\x  \nonumber  \\ 
&-& \sum_{i=1}^n  \alpha_i {\mu'_i} \int \frac{f(\x) p_i(\x) p_j(\x)}{(\sum_{k=1}^n \alpha_k p_k(\x))^2}d\x \nonumber \\
&= &    {\sigma'}_j^2 +  {\mu'}_j^2 
- \sum_{i=1}^n  \alpha_i {\mu'_i} \int \frac{f(\x) p_i(\x) p_j(\x)}{(\sum_{k=1}^n \alpha_k p_k(\x))^2}d\x, \nonumber
\end{eqnarray}
and thus
 \begin{eqnarray}
\label{betafixedproof10}
{\mu'}_j^2 =  \sum_{i=1}^n  \alpha_i {\mu'_i} \int \frac{f(\x) p_i(\x) p_j(\x)}{(\sum_{k=1}^n \alpha_k p_k(\x))^2}d\x.
\end{eqnarray}
Eq.~\eqref{betafixedproof10} holds when for all $i$, all ${\mu'_i}$ values are equal. Thus Eq.~\eqref{betafixedproof8} is filled  when for all $i$, $\alpha_i = \beta_i$ and all ${\mu'_i}$ are equal. This is in concordance with Theorems 9.2 and 9.4 of Veach's thesis \cite{Veach:PhD}, which deal with the optimality of the deterministic mixture estimator (Theorem 9.2) and the random mixture estimator (Theorem 9.4) MIS estimator, see Section~\ref{preliminaries}. When all ${\mu'_i}$ are equal the variances of both estimators are the same (see see Section~\ref{preliminaries}), and being the optimal by Theorem 9.4 the balance heuristic estimator, it implies that for all $i$, $\alpha_i = \beta_i$.

\bibliographystyle{Perfect}
%\bibliography{missampling_cropped}

\newpage
\begin{landscape}

\begin{table*}
\caption{ {We show the metric $E_F^{-1}=V[F]\cdot Cost[F]$ and $E_G^{-1}=V[G]\cdot Cost[G]$, i.e., the product of variance and cost for the $F$ estimator and for the optimal $G$ estimator, using both the same $\alpha_k$ values.
  For the 5 numerical examples using equal count of
  samples, count inversely proportional to the variances of
  independent estimators~\cite{vrcaiHavranSbert14},\cite{SbertHavran2017},
 and for the three estimators
  defined in \cite{Sbert2018}. The sampling costs are (1, 6.24, 3.28). In Example 5, we present the case with equal costs  (1,1), and different costs (1,5).} }
  \label{table:functions1}
  {\scriptsize
\begin{tabular}{ |c|c|c|c|c|c|c|c|c|c|c|c|c|c|c|c|c|}
 \hline
 & \multicolumn{2}{|c|}{{ \multirow{ 2}{*}{ \bf Ex. 1} }} & \multicolumn{2}{|c|}{{ \multirow{ 2}{*}{ \bf Ex. 2} }} & \multicolumn{2}{c}{{ \multirow{ 2}{*}{ \bf Ex. 3} }} & \multicolumn{2}{|c|}{{ \multirow{ 2}{*}{ \bf Ex. 4} }} & \multicolumn{2}{|c|}{{\bf Ex. 5}} & \multicolumn{2}{|c|}{{\bf Ex. 5}} \\
 \cline{10-13}
 & \multicolumn{2}{|c|}{  } & \multicolumn{2}{|c|}{  } & \multicolumn{2}{|c|}{  } & \multicolumn{2}{|c|}{  } & costs=(1,1)  &  costs=(1,1) & costs=(1,5)  &  costs=(1,5)  \\
 \cline{1-13}
 {\bf Estimator } & F& G &  F& G &  F& G & F& G   &  F & G &  F & G  \\
  \hline
  \hline
 $\alpha_k \propto \frac{1}{n}$   &102.26 &89.40 &17.24&15.44 &37.47 &31.80 &98.68&83.78  &0.28 &0.23 &0.83&0.40   \\
  $\alpha_k \propto \frac{1}{c_k v_k}$   \cite{vrcaiHavranSbert14}%,\cite{SbertHavran2017}    
  &49.53 &41.29 &9.28 &8.10  & 4.03 &3.85  & 300.12 &294.85  &0.31 &0.26 &2.76 &2.33  \\
     $\alpha_k \propto \frac{1}{c_k m^2_k}$   \cite{Sbert2018}   &54.36 &46.2 &9.82&8.49  & 3.12 &2.49 & 534.37 &449.33  &0.20 &0.15 &1.51 &1.00 \\
     $\alpha_k \propto \frac{\sigma_{k,eq}}{\sqrt{c_k}}$  \cite{Sbert2018}   &81.43 &69.88 &13.54 &11.67 &28.68  &23.17 & 91.01 &73.54  &1.00 &0.98 &2.90 &2.50  \\
   $\alpha_k \propto \frac{M_{k,eq}}{\sqrt{c_k}}$  \cite{Sbert2018} &79.73 &67.77 &13.08 &11.35 &25.74  &20.76 &31.77 &25.90  &0.29 &0.24 &2.72 &2.28  \\
 \hline
\end{tabular}
}

\end{table*}
\end{landscape}

\end{document}